\documentclass[11pt,final]{article}

\usepackage{amsmath}
\usepackage{amsthm}
\usepackage{amssymb}
\usepackage[pdftex]{graphicx,color}
\usepackage[linesnumbered,ruled,vlined]{algorithm2e}

\numberwithin{equation}{section}

\theoremstyle{definition}\newtheorem{definition}{Definition}
\theoremstyle{plain}\newtheorem{theorem}{Theorem}
\theoremstyle{plain}
\theoremstyle{plain}\newtheorem{corollary}{Corollary}
\theoremstyle{plain}\newtheorem{lemma}{Lemma}
\theoremstyle{definition}
\theoremstyle{definition}
\theoremstyle{definition}
\theoremstyle{definition}

\newcommand{\T}{\mathrm{T}}
\newcommand{\bbR}{\mathbb{R}}

\begin{document}

\title{A simple linear time algorithm for smallest enclosing circles on the (hemi)sphere}

\author{
\textsc{Jens Flemming}%
\footnote{Zwickau University of Applied Sciences, Faculty of Physical Engineering/Computer Sciences , D-08012 Zwickau, Germany,
jens.flemming@fh-zwickau.de.}
}

\date{\today
}

\maketitle

\begin{abstract}
Based on Welzl's algorithm for smallest circles and spheres we develop a simple linear time algorithm for finding the smallest circle enclosing a point cloud on a sphere. The algorithm yields correct results as long as the point cloud is contained in a hemisphere, but the hemisphere does not have to be known in advance and the algorithm automatically detects whether the hemisphere assumption is met.
\par
For the full-sphere case, that is, if the point cloud is not contained in a hemisphere, we provide hints on how to adapt existing linearithmic time algorithms for spherical Voronoi diagrams to find the smallest enclosing circle.
\end{abstract}

\section{Problem statement and applications}

Given $n$ points on a sphere in $\bbR^3$ we want to find a smallest circle on the sphere enclosing all points. This is a well-known problem in the plane, see \cite{Wel91} for some early references, and has been studied on the sphere at least since the 1980s, see \cite{DreWes83} and references therein.

In operations research the smallest circle problem appears in facility planning. How to place a facility to minimize the distance to the most remote of several demand points? In this context the smallest circle problem is known as 1-center problem or minimax location problem. It may be considered in the plane or on the sphere, where the latter is common in case of large areas to cover.

Given a region of interest, in cartography one may ask for a map projection minimizing maximum distortion in that region. In case of azimuthal equidistant projections optimal mapping parameters are given by the solution of the smallest circle problem with respect to the vertices of the region's boundary polygon.

Taking into account the obvious fact that on a sphere the smallest enclosing circle is equivalent to the largest empty circle (see \cite[Theorem~1]{DreWes83} for a formal proof) several other applications arise. For instance, the then so called maximin location problem is solved in operations research to find locations for facilities as far away from a set of points as possible. The facility could be a toxic waste dump, a new restaurant avoiding too much competition, or a military site as far away from enemy positions as possible.

In geography the largest empty circle problem is solved to determine so called poles of inaccessibility, that is, locations farthest from some topographical feature (landmass or sea, for instance). 

Especially for geodata applications fast algorithms for solving the smallest circle problem on the sphere or, equivalently, the largest empty circle problem are of importance, because geodata usually comes as large point clouds or polygons with lots of vertices. For instance, in the OpenStreetMap data base \cite{osm} a national border may be a polygon with more than 100\,000 vertices (France 101\,751, Germany 159\,625, \cite{osmdata}). Considering larger regions may result in millions of points to process.

\section{Algorithms in the literature}

To our best knowledge the oldest attempt to solve the smallest circle problem on a sphere is \cite{DreWes83} by Drezner and Wesolowsky. There the problem is approximately solved based on steepest descent for a minimization problem. This yields local minima only. Thus, a relatively complex procedure is proposed in \cite{DreWes83} to obtain global solutions from local ones. There is no time complexity analysis, but numerical results in \cite{DreWes83} indicate that time complexity is worse than linear. Solution accuracy seems to be in the order of 5 to 6 decimal digits.

In \cite{Pat95} Patel derives a system of nonlinear equations from the KKT optimality conditions of a minimization problem equivalent to the smallest circle problem and suggests a concrete numerical method for approximately solving that system. Regarding the restriction of our algorithm to point clouds contained in a hemisphere (cf.\ below) it's noteworthy that \cite{Pat95} discusses simplifications for a special hemispherical case.

Xue and Sun for the first time propose a linear time algorithm for the spherical smallest circle problem in \cite{XueSun95}. Their algorithm is exact, that is, no numerical approximation techniques are used. They formulate a quadratic optimization problem equivalent to the smallest circle problem and show that there exists a uniquely determined minimizer if the point cloud is contained in a hemisphere. For solving that quadratic problem in linear time they sketch some reformulations and refer to the general approach of Megiddo in \cite{Meg84} (and also to Dyer \cite{Dye86}), who for quadratic problems then refers to \cite{Meg83}. How to handle or even detect the full-sphere case is not considered. They only show that for the full-sphere case the smallest circle problem may have several solutions. Our algorithm proposed below will be based on the ideas of Megiddo, too. But we will give a simple concrete algorithm, which in addition is able to automatically detect the full-sphere case.

The authors of \cite{DasChaCha99} present a concrete exact algorithm for solving the smallest circle problem on a hemisphere in quadratic time. Obviously they are not aware of the linear time solution in \cite{XueSun95} published several years before, although the introduction of \cite{DasChaCha99} contains many valuable references to previous approaches to the problem. In \cite{DasChaCha01} same authors extend their ideas to point clouds not contained in a hemisphere. Time complexity then is cubic.

\section{The algorithm}

At the end of this section we state our algorithm for finding the smallest enclosing circle of a point cloud on a sphere assuming the point cloud is contained in an unknown hemisphere. We prepare this result by recalling Welzl's algorithm and by discussing a tight relation between smallest spherical circles and smallest spheres in 3d space. Based on this discussion we derive a simple condition to check whether we left the hemisphere while processing the point cloud point by point.

Ingredients and references allowing for an almost linear (that is, linearithmic) time full-sphere algorithm will be given in the next section.

\subsection{Welzl's algorithm for planar circles and spheres}

Based on linear time algorithms for more general linear programming tasks Welzl developed a very simple linear time algorithm for the planar smallest circle problem in \cite{Wel91}. Welzl's algorithm easily extends to higher dimensions, especially to smallest enclosing spheres in 3d space.

In \cite{Wel91} Welzl's algorithm is given in a recursive manner with recursion depth equal to the number of points to process. For large point clouds such recursive implementation may be infeasible due to limited stack size or otherwise limited recursion depth. For instance, default maximum recursion depth for interpreters of the widely used Python programming language ranges from 10 for Python on some microcontrollers to 1000 in standard Python\footnote{See https://docs.python.org/3/library/sys.html\#sys.getrecursionlimit for information on how to obtain these values.}. Here we give Welzl's algorithm in an almost iterative formulation, where recursive calls only occur if a point is identified as boundary point of the smallest circle or sphere enclosing all points processed so far.

The algorithm or the corresponding function \texttt{welzl} takes two arguments: the list $P$ of points to enclose and a set $B$ of points known to lie on the smallest enclosing circle's or sphere's boundary. Given at least two points $x_1,\ldots,x_n$ in $\bbR^d$ with $d\in\{2,3\}$ start the algorithm by calling the \texttt{welzl} function with arguments $P:=(x_1,\ldots,x_n)$ and $B:=\emptyset$ (no boundary points known in advance). In the algorithm we always write `circle'. For $d=3$ `sphere` would be more appropriate.

\begin{algorithm}[H]
\caption{Welzl's algorithm in almost iterative form}
\DontPrintSemicolon
\SetArgSty{}
\SetFuncArgSty{}
\SetKwProg{Func}{function}{}{}
\SetKwFunction{Welzl}{welzl}
\SetKw{Return}{return}
\SetKwIF{If}{ElseIf}{Else}{if}{then}{else if}{else}{}

\Func{\Welzl{$P$, $B$}}{
    \eIf{$|B|=d+1$}{\label{line:Bd1}
        \Return{circle determined by points in $B$}\label{line:circle_d1}
    }{
        $m$ $\leftarrow$ $\max\{0,2-|B|\}$\;
        $C$ $\leftarrow$ smallest circle enclosing $B$ and first $m$ points in $P$\;\label{line:circle_d}
        \For{$i=m+1,\ldots,|P|$}{
            \If{$i$th point of $P$ not enclosed by $C$}{
                $b$ $\leftarrow$ $i$th point of $P$\;
                $C$ $\leftarrow$ \Welzl{first $i-1$ points of $P$, $B\cup\{b\}$}\;\label{line:recursion}
                Move $i$th point of $P$ to front (make it the first point).\;\label{line:front}
            }
        }
        \Return{C}
    }
}
\end{algorithm}

The circle/sphere in line~\ref{line:circle_d1} of the algorithm is uniquely determined by the points in $B$ if points aren't collinear ($d=2$) or cocircular ($d=3$), which is always true here by the construction of $B$. In line~\ref{line:circle_d} the smallest circle enclosing two points or the smallest sphere enclosing two or three points has to be computed, which is trivial. Recursion depth in line~\ref{line:recursion} is at most $d+1$. Line~\ref{line:front} implements the move-to-front heuristic suggested in \cite{Wel91} to speed up computations by constructing large initial circles in line~\ref{line:circle_d}.

In its original form Welzl's algorithm has to be randomized to obtain linear expected runtime. But in \cite[Section~3]{Wel91} Welzl also shows that with the variant given above it's sufficient to randomly permutate the points only once before running the algorithm. This permutation step is not stated explicitly in the algorithm above.

\subsection{Smallest spherical circles vs.\ smallest spheres}

Although we already used the term (spherical) circle above, here we give a precise definition for the sake of mathematical rigor. Note that throughout this article we assume that the point cloud to enclose lives on a sphere with radius 1 centered at the origin. This sphere will be referred to as main sphere if there is a risk of ambiguity.

\begin{definition}
Given a center point $c$ on the sphere and a radius $r\in(0,\pi)$ the corresponding \emph{spherical circle} is the set of all points on the sphere with geodesic distance $r$ to the center $c$.
\end{definition}

A spherical circle devides the sphere into two caps, a smaller one and a larger one. For radius $r=\frac{\pi}{2}$ both caps are of equal size. The smaller cap can be represented as the intersection of a ball and the sphere:

\begin{lemma}\label{th:intersection}
The small cap defined by a spherical circle with radius $r<\frac{\pi}{2}$ and center $c$ equals the intersection of the sphere and the ball with radius $\sin c$ centered at $(\cos r)\,c$.
\end{lemma}

\begin{proof}
This follows immediately from basic trigonometry.
\end{proof}

\begin{definition}
A set of points on the sphere is \emph{contained in a hemisphere} if there is an enclosing spherical circle with radius $r<\frac{\pi}{2}$.
\end{definition}

If a point cloud is contained in a hemisphere, the point cloud's smallest enclosing circle is uniquely determined (see \cite[Theorem~3]{XueSun95} or \cite[Lemma~1]{Wel91} for the planar case, the proof for easily extends to the spherical setting). If the point cloud is not contained in a hemisphere, there might be several smallest enclosing circles (see \cite[first paragraph in Section~2]{XueSun95}).

The circle in Lemma~\ref{th:intersection} above is a great circle on the corresponding ball's surface. Consequently, there's no smaller ball satisfying that intersection property. This observation can be extended as follows:

\begin{lemma}\label{th:rrcc}
Let the point cloud be contained in a hemisphere. If the smallest enclosing circle on the sphere is centered at $c$ with radius $r$ and if the point cloud's smallest enclosing sphere is centered at $\tilde{c}$ with radius $\tilde{r}$, then
\begin{equation}
\tilde{c}=(\cos r)\,c\qquad\text{and}\qquad\tilde{r}=\sin r.
\end{equation}
If the point cloud is not contained in a hemisphere, then the smallest enclosing sphere is centered at the origin with radius 1.
\end{lemma}
\begin{proof}
The assertions are quite obvious. Given the smallest enclosing circle, take corresponding ball from Lemma~\ref{th:intersection}. Its surface is the smallest enclosing sphere because any smaller enclosing sphere would yield a smaller enclosing circle via intersection with the main sphere. The other way round, starting with the smallest enclosing sphere, the circle defined by intersection with the main sphere has to be the smallest enclosing circle. Else, it would give rise to a smaller enclosing sphere via Lemma~\ref{th:intersection}.
\par
Now assume that the point cloud is not contained in a hemisphere. If the smallest enclosing sphere would have radius $\tilde{r}<1$, by intersection with the main sphere we would obtain an enclosing circle with radius \begin{equation}
r=\arcsin\tilde{r}<\frac{\pi}{2},
\end{equation}
which contradicts the assumption $r\geq\frac{\pi}{2}$.
\end{proof}

The basic idea of Lemma~\ref{th:rrcc} that smallest enclosing circles on a sphere are closely related to smallest enclosing spheres has already been mentioned in \cite{stackex} without any details or proofs.

\subsection{Derivation of the algorithm}

The idea of our Welzl-type algorithm for point clouds on the sphere is to apply Welzl's algorithm in 3d (smallest enclosing sphere), but implement some modifications which on the one hand allow to detect whether the point cloud is contained in a hemisphere and on the other hand bring the algorithm very close to the 2d variant of Welzl's algorithm.

Our algorithm heavily relies on the following result:

\begin{theorem}\label{th:b4}
Let $d=3$ in Welzl's algorithm and start the algorithm with some point cloud $P$ and an empty set $B$.
\begin{itemize}
\item[(i)]
If the point cloud $P$ is contained in a hemisphere, then the condition $|B|=4$ in line~\ref{line:Bd1} of the algorithm is never satisfied.
\item[(ii)]
If the point cloud $P$ is not contained in a hemisphere, then at least one of the following states will be observed during execution of Welzl's algorithm:
\begin{itemize}
\item[(a)]
In line~\ref{line:circle_d} of the algorithm the initial sphere is constructed from two points which lie antipodally on the main sphere.
\item[(b)]
The \texttt{welzl} function is called with $|B|=3$ and the three boundary points lie on a great circle of the main sphere.
\item[(c)]
The condition $|B|=4$ in line~\ref{line:Bd1} of the algorithm is satisfied.
\end{itemize}
\end{itemize}
\end{theorem}
\begin{proof}
To prove the theorem we take a close look at Welzl's algorithm. Welzl's algorithm processes the point cloud point by point. It starts with two points (even in the 3d case considered here), determines the smallest enclosing sphere for these two points, then goes on to the third point. If the third point is enclosed by the initial sphere, the fourth point is processed. If the third point is not enclosed by the initial sphere, the point will be marked as boundary point. A recursive call to the \texttt{welzl} function is used to determine a new sphere enclosing the first two points and having the third point on it (that is, on the boundary of corresponding ball). Then the fourth point is processed in the same way as the third point, and so on.
\par
During recursive calls to \texttt{welzl} another boundary point may be detected, leading to a recursive call again. This way the smallest enclosing sphere is determined point by point with recursive calls of \texttt{welzl} whenever the current sphere turns out to be too small. Recursion depth is limited by the number of boundary points uniquely determining a sphere, that is, by 4.
\par
To prove (i) observe that in a depth-3 recursive call to \texttt{welzl} there are three boundary points ($|B|=3$) and a number of points in $P$ to enclose. On the one hand, the smallest sphere defined by the boundary points has all three boundary points on one great circle. On the other hand, from Lemma~\ref{th:rrcc} we know that the smallest enclosing sphere for $P$ with $B$ as boundary points is a sphere with the three boundary points on a great circle, too. Thus, both spheres coincide and the smallest enclosing sphere for $P$ with boundary points $B$ can be determined from $B$ without touching $P$. All points of $P$ automatically are enclosed by the smallest enclosing sphere of $B$, that is, there will be no further recursive call of \texttt{welzl}.
\par
To prove (ii) assume that the point cloud $\{x_1,\ldots,x_n\}$ to enclose is not contained in a hemisphere. Then there is $i$ such that $\{x_1,\ldots,x_{i-1}\}$ is contained in a hemisphere, but $\{x_1,\ldots,x_i\}$ is not. While processing $x_i$, the point will be marked as boundary point and \texttt{welzl} will be called with $P=(x_1,\ldots,x_{i-1})$ and $B=\{x_i\}$. During this recursive call the smallest enclosing sphere for $x_i$ and some other point $x_j$ with $j\in\{1,\ldots,i-1\}$ will be constructed. If state (a) is not observed, there will be some $x_k$, $k\in\{1,\ldots,i-1\}$ not enclosed by this sphere, because the sphere's radius is strictly below 1 and $\{x_1,\ldots,x_i\}$ is not contained in any hemisphere. Thus, \texttt{welzl} will be called recursively with $P=(x_1,\ldots,x_{k-1},x_{k+1},x_{i-1})$ and $B=\{x_i,x_k\}$. Following the same reasoning again we see a third recursive call to \texttt{welzl}, now with $|B|=3$. If state (b) is not observed in this call, analogous reasoning applies again, resulting in a fourth recursive call, now with $|B|=4$, which proves the theorem.
\end{proof}

The theorem states that up to some easily detectable edge cases (cf.\ states (a) and (b) in the theorem) the condition $|B|=4$ can be checked to determine whether the point cloud lies in a hemisphere or not. As long as the points being processed are contained in a hemisphere the 3d algorithm is almost identical to the 2d algorithm due to the lack of depth-4 recursive calls. The only difference is that for $|B|=3$ the 2d variant immediately returns the circle determined by the boundary points whereas the 3d variant also checks whether all points processed so far are enclosed by the smallest sphere determined by $B$. If we know in advance that all points are contained in a hemisphere we may skip this additional check, which shows that the 2d Welzl algorithm correctly works for point clouds on hemispheres, too, in addition to planar clouds.

\subsection{Statement of the algorithm}

Here we state the complete algorithm based on the considerations from the previous subsections. The only difference to the derivation above is that to avoid a fourth recursive call we perform the hemisphere check immediately after constructing the smallest enclosing sphere for three boundary points. This does not change the algorithm's behavior, but only its structure.

Like for Welzl's algorithm, start the following algorithm with $P$ containing a list of all points to enclose and $B=\emptyset$. Randomly permutate the list of points before applying the algorithm.

\begin{algorithm}[H]
\caption{Smallest enclosing circle on the sphere}
\DontPrintSemicolon
\SetArgSty{}
\SetFuncArgSty{}
\SetKwProg{Func}{function}{}{}
\SetKwFunction{Secots}{secots}
\SetKw{Return}{return}
\SetKw{Stop}{stop}
\SetKwIF{If}{ElseIf}{Else}{if}{then}{else if}{else}{}

\Func{\Secots{$P$, $B$}}{
    \eIf{$|B|=3$}{
        $C$ $\leftarrow$ circle determined by points in $B$\;\label{line:c3}
        \If{$C$ is a great circle}{\label{line:state_a}
            \Stop\tcc*{points not in hemisphere, state (a)}
        }
        \If{$C$ does not encloses all points in $P$}{\label{line:hemitest}
            \Stop\tcc*{points not in hemisphere, state (c)}
        }
        \Return{$C$}\;
    }{
        $M$ $\leftarrow$ union of $B$ and $2-|B|$ first points in $P$\;
        \If{points in $M$ are antipodal on main sphere}{
            \Stop\tcc*{points not in hemisphere, state (b)}
        }
        $C$ $\leftarrow$ smallest circle enclosing $M$\;
        \For{$i=2-|B|+1,\ldots,|P|$}{
            \If{$i$th point of $P$ not enclosed by $C$}{
                $b$ $\leftarrow$ $i$th point of $P$\;
                $C$ $\leftarrow$ \Secots{first $i-1$ points of $P$, $B\cup\{b\}$}\;\label{line:call}
                Move $i$th point of $P$ to front (make it the first point).\;\label{line:mtf}
            }
        }
        \Return{C}
    }
}
\end{algorithm}

The stop command in the algorithm stops the whole program, aborting all recursive calls to \texttt{secots}. States (a), (b), (c) refer to Theorem~\ref{th:b4}.

\begin{corollary}
The algorithm has linear expected runtime with respect to the number of points to process.
\end{corollary}
\begin{proof}
From the derivation of the algorithm we see that the runtime is not worse than the runtime of Welzl's algorithm for smallest enclosing sphere's in 3d space. Thus, runtime is linear.
\end{proof}

\subsection{Implementation details}

For testing whether a point is enclosed by the current circle we do not have to calculate the point's distance to the circle's center. Instead we may test whether the point lies on the correct side of the plane containing the circle. If the current circle has center $c$ and radius $r$, then this plane is
\begin{equation}
\bigl\{x\in\bbR^3:\,u^\T\,x=t\bigr\},\qquad\text{with \quad$u:=c$\quad and \quad$t:=\cos r$.}
\end{equation}
A point $x$ is on the correct side if $u^\T\,x\geq t$. This test requires fewer elementary operations than a usual containment test for a circle. Given the plane, that is, $u$ and $t$, the circle's center is $u$ and the radius is $\arccos t$.

Working with such planes, construction of the circle in line~\ref{line:c3} of our algorithm boils down to solving the system of linear equations
\begin{align}
b_1^\T\,v&=1,\notag\\
b_2^\T\,v&=1,\\
b_3^\T\,v&=1\notag
\end{align}
for $v\in\bbR^3$, where $b_1,b_2,b_3$ are the three boundary points. Then the plane (and, thus, the circle) is determined by $u:=\frac{v}{\Vert v\Vert}$ and $t:=\frac{1}{\Vert v\Vert}$. The circle is a great circle (cf.\ line~\ref{line:state_a}) if and only if the system has no solution.

To save computation time we may skip the hemisphere test in line~\ref{line:hemitest} if we know in advance that the point cloud is contained in a hemisphere.

Whenever a recursive call to \texttt{secots} returns in line~\ref{line:call} (that is, no stop) we know that the points processed so far are contained in a hemisphere. If the next point $x$ not enclosed by the current circle satisfies $u^\T\,x>-t$, from simple geometric considerations we see that the already processed points and $x$ together are contained in a hemisphere, too. Passing this knowledge on to the subsequent recursive call of \texttt{secots} avoids the then unnecessary hemisphere test in line~\ref{line:hemitest}.

Care has to be taken implementing the move-to-front heuristic in line~\ref{line:mtf}. This operation has to be executed in constant time, which is possible if the point cloud is stored in a doubly linked list.

A ready-to-use Python implementation along the lines sketched here is available in \cite{github}. Figure~\ref{fig:time} has been obtained with this implementation and shows that runtime indeed behaves linearly as expected from the theory.

\begin{figure}[ht]
\includegraphics[width=1\textwidth]{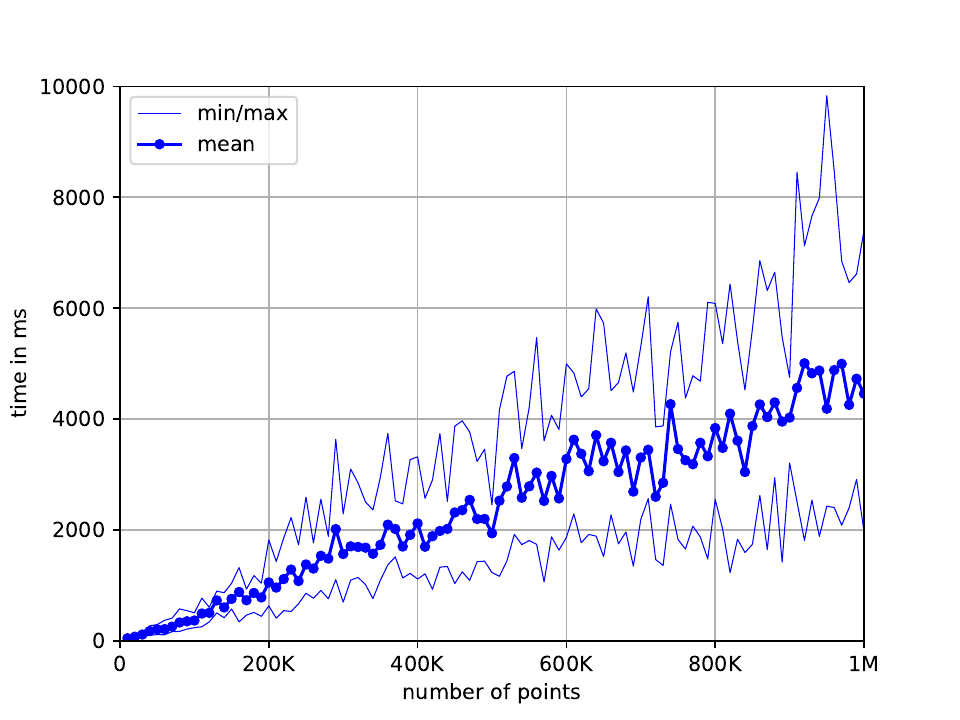}
\caption[Numerical runtime]{\label{fig:time}Runtime for point clouds of different sizes between 100\,000 and 1~million points. Points are uniformly distributed in a `rectangle'. Longitudes span 90 degrees, latitudes span 60 degrees. For each cloud size the algorithm is run 10 times. Computations were done on one core of an Intel Core i7-8550U CPU with 1.8\,GHz. Code is available in \cite[example `time']{github}.}
\end{figure}

\section{The full-sphere case}

As our linear time algorithm only yields results for point clouds contained in a hemisphere one may ask whether there are linear time algorithms for the full-sphere case. To our best knowledge the answer is `no'. Solutions in the full-sphere case aren't uniquely determined. Thus, processing the point cloud point by point would yield kind of local solution only.

We mentioned in the introductory section that the smallest enclosing circle problem on the sphere is equivalent to the largest empty circle problem on the sphere. The planar largest empty circle problem typically is solved by computing the Voronoi diagram of the point cloud and then looking for the Voronoi vertex with the largest distance to its cells' centers, see \cite{Tou83}, for instance. This vertex is the center of the largest empty circle. Here a Voronoi vertex is a point where three or more Voronoi cells meet.

To solve the largest enclosing circle problem on the sphere we have to compute a Voronoi diagram on the sphere. This is possible in linearithmic time. Concrete algorithms are given in \cite{DinMam10} and \cite{ZheEnnRicPal11}. If the point cloud has $n$ points, the Voronoi diagram has $2\,n-4$ vertices, see \cite[Section~2]{ZheEnnRicPal11}. Thus, the vertex with largest distance to its cells' centers can be found in linear time.

Note that the algorithms presented in \cite{DinMam10} and \cite{ZheEnnRicPal11} next to the Voronoi vertices yield all relevant information like neighboring cells and distances to cell centers within their linearithmic time bound. Thus, the overall algorithm for computing the largest empty circle has linearithmic runtime, too.

\section{Conclusions}

We have shown that computing smallest enclosing circles for point clouds on the sphere is possible in linear time by mixing the 2d and the 3d variant of Welzl's algorithm. Although our algorithm only yields correct results if the point cloud is contained in a hemisphere, the hemisphere does not have to be known in advance. In addition, our algorithm detects whether the hemisphere assumption is true or not.

Algorithms developed during the past 40 years starting with  \cite{DreWes83} have polynomial runtime or yield approximate solutions only. Especially for geodata applications those algorithms are too slow. With the algorithm presented in this article millions of points can be processed on consumer hardware within seconds.

Next to the theoretical linear time result, which directly carries over from Welzl's work to ours, we have provided numerical evidence that our implementation does not violate the theoretical limit.

A ready-to-use Python implementation of our algorithm published in \cite{github} makes our results easily accessible to other researchers and practitioners.

\bibliography{secots}

\begin{thebibliography}{10}

\bibitem{stackex}
J.~Brower.
\newblock Minimum bounding circle on a sphere.
\newblock Mathematics Stack Exchange.
\newblock Comment by user Henry, https://math.stackexchange.com/q/4336094
  (version: 2021-12-17).

\bibitem{DasChaCha99}
P.~Das, N.~R. Chakraborti, and P.~K. Chaudhuri.
\newblock A polynomial time algorithm for a hemispherical minimax location
  problem.
\newblock {\em Operations Research Letters}, 24(1):57--63, 1999.

\bibitem{DasChaCha01}
P.~Das, N.~R. Chakraborti, and P.~K. Chaudhuri.
\newblock Spherical minimax location problem.
\newblock {\em Computational Optimization and Applications}, 18(3):311--326,
  2001.

\bibitem{DinMam10}
J.~Dinis and M.~Mamede.
\newblock Sweeping the sphere.
\newblock In {\em 2010 International Symposium on Voronoi Diagrams in Science
  and Engineering}, pages 151--160, 2010.

\bibitem{DreWes83}
Z.~Drezner and G.~O. Wesolowsky.
\newblock Minimax and maximin facility location problems on a sphere.
\newblock {\em Naval Research Logistics Quarterly}, 30(2):305--312, 1983.

\bibitem{Dye86}
M.~E. Dyer.
\newblock On a multidimensional search technique and its application to the
  euclidean one-centre problem.
\newblock {\em SIAM Journal on Computing}, 15(3):725--738, 1986.

\bibitem{github}
J.~Flemming.
\newblock {Secots -- Smallest enclosing circles on the sphere}.
\newblock https://github.com/jeflem/secots, 2024.

\bibitem{Meg83}
N.~Megiddo.
\newblock Linear-time algorithms for linear programming in {$R^3$} and related
  problems.
\newblock {\em SIAM Journal on Computing}, 12(4):759--776, 1983.

\bibitem{Meg84}
N.~Megiddo.
\newblock Linear programming in linear time when the dimension is fixed.
\newblock {\em J. ACM}, 31(1):114–--127, 1984.

\bibitem{osmdata}
{OpenStreetMap contributors}.
\newblock {Planet dump retrieved from https://planet.osm.org, 2024-07-23}.

\bibitem{osm}
{OpenStreetMap contributors}.
\newblock {OpenStreetMap data base}.
\newblock https://www.openstreetmap.org, 2024.

\bibitem{Pat95}
M.~Patel.
\newblock Spherical minimax location problem using the euclidean norm:
  Formulation and optimization.
\newblock {\em Computational Optimization and Applications}, 4:79--90, 1995.

\bibitem{Tou83}
G.~T. Toussaint.
\newblock Computing largest empty circles with location constraints.
\newblock {\em International Journal of Computer \& Information Sciences},
  12(5):347--358, 1983.

\bibitem{Wel91}
E.~Welzl.
\newblock Smallest enclosing disks (balls and ellipsoids).
\newblock In Hermann Maurer, editor, {\em New Results and New Trends in
  Computer Science}, pages 359--370. Springer Berlin Heidelberg, 1991.

\bibitem{XueSun95}
G.~Xue and S.~Sun.
\newblock The spherical one-center problem.
\newblock In D.-Z. Du and P.~M. Pardalos, editors, {\em Minimax and
  Applications}, pages 153--156. Springer US, Boston, MA, 1995.

\bibitem{ZheEnnRicPal11}
X.~Zheng, R.~Ennis, G.~P. Richards, and P.~Palffy-Muhoray.
\newblock A plane sweep algorithm for the voronoi tessellation of the sphere.
\newblock {\em Electronic-Liquid Crystal Communications}, 2011.

\end{thebibliography}

\end{document}